\documentclass[a4paper]{article}

\usepackage[pages=all, color=black, position={current page.south}, placement=bottom, scale=1, opacity=1, ' vshift=5mm]{background}
\SetBgContents{
}      

\usepackage[margin=1in]{geometry} 

\usepackage{amsmath}
\usepackage{amsthm}
\usepackage{amssymb}
\usepackage[T1]{fontenc}  
\usepackage[utf8]{inputenc}
\usepackage{hyperref}
\hypersetup{
	unicode,
	pdfauthor={Author One, Author Two, Author Three},
	pdftitle={A simple article template},
	pdfsubject={A simple article template},
	pdfkeywords={article, template, simple},
	pdfproducer={LaTeX},
	pdfcreator={pdflatex}
}

\usepackage[sort&compress,numbers,square]{natbib}
\theoremstyle{plain}
\newtheorem{theorem}{Theorem}
\newtheorem{corollary}{Corollary}
\newtheorem{lemma}{Lemma}
\newtheorem{observation}{Observation}
\newtheorem{claim}{Claim}

\newtheorem{proposition}[theorem]{Proposition}

\theoremstyle{definition}
\newtheorem{definition}{Definition}

\newtheorem{problem}[theorem]{Problem}

\usepackage{graphicx, color}
\graphicspath{{fig/}}

\usepackage{algorithm, algpseudocode} 
\usepackage{mathrsfs} 

\usepackage{lipsum}

\title{Line-Constrained $k$-Semi-Obnoxious Facility Location}
\author{Vishwanath R. Singireddy$^1$ \and Manjanna Basappa$^1$\thanks{The authors, M. Basappa and N. R. Aravind, were partially supported by SERB TARE Grant TAR/2022/000397.} \and N. R. Aravind$^{2\hspace{0.3mm}*}$}

\date{
	$^1$BITS Pilani, Hyderabad Campus, India \\ \texttt{\{p20190420, manjanna\}@hyderabad.bits-pilani.ac.in}\\%
	$^2$IIT Hyderabad, India\\ \texttt{aravind@cse.iith.ac.in}\\[2ex]%
}

\begin{document}
	\maketitle
	
	\begin{abstract}
		Suppose we are given a set $\cal B$ of blue points and a set $\cal R$ of red points, all lying above a horizontal line $\ell$, in the plane. Let the weight of a given point $p_i\in {\cal B}\cup{\cal R}$ be $w_i>0$ if $p_i\in {\cal B}$ and $w_i<0$ if $p_i\in {\cal R}$, $|{\cal B}\cup{\cal R}|=n$, and $d^0$($=d\setminus\partial d$) be the interior of any geometric object $d$. We wish to pack $k$ non-overlapping congruent disks $d_1$, $d_2$, \ldots, $d_k$ of minimum radius, centered on $\ell$ such that $\sum\limits_{j=1}^k\sum\limits_{\{i:\exists p_i\in{\cal R}, p_i\in d_j^0\}}w_i+\sum\limits_{j=1}^k\sum\limits_{\{i:\exists p_i\in{\cal B}, p_i\in d_j\}}w_i$ is maximized, i.e., the sum of the weights of the points covered by $\bigcup\limits_{j=1}^kd_j$ is maximized. Here, the disks are the obnoxious or undesirable facilities generating nuisance or damage (with quantity equal to $w_i$) to every demand point (e.g., population center) $p_i\in {\cal R}$ lying in their interior. In contrast, they are the desirable facilities giving service (equal to $w_i$) to every demand point $p_i\in {\cal B}$ covered by them. The line $\ell$ represents a straight highway or railway line. These $k$ semi-obnoxious facilities need to be established on $\ell$ to receive the largest possible overall service for the nearby attractive demand points while causing minimum damage to the nearby repelling demand points. We show that the problem can be solved optimally in $O(n^4k^2)$ time. Subsequently, we improve the running time to $O(n^3k \cdot\max{(n, k)})$. Furthermore, we addressed two special cases of the problem where points do not have arbitrary weights. In the first case, the objective is to encompass the maximum number of blue points while avoiding red points. The second case aims to encompass all the blue points with the minimum number of red points covered. We show that these two special cases can be solved in $O(n^3k\cdot\max{(\log{n},k)})$ time. For the first case, when $k=1$, we also provide an algorithm that solves the problem in $O(n^3)$ time, and subsequently, we improve this result to $O(n^2\log{n})$. For the latter case, we give $O(n\log{n})$ time algorithm that uses the farthest point Voronoi diagram. The above-weighted variation of locating $k$ semi-obnoxious facilities may generalize the problem that Bereg et al. (2015) studied where $k=1$ i.e., the smallest radius maximum weight circle is to be centered on a line. Furthermore, we consider a generalization of the weighted problem where we are given $t$ horizontal lines instead of one line. We give an $O(n^4k^2t^5)$ time algorithm for this problem. Finally, we consider a discrete variant where a set of $s$ candidate sites (in convex position) for placing $k$ facilities is pre-given ($k<s$). We propose an algorithm that runs in $O(n^2s^2+ns^5k^2)$ time for this discrete variant.
  
  
		\noindent\textbf{Keywords:} Semi-obnoxious Facility Location, Complete Weighted Directed Acyclic Graph, Minimum-weight $k$-link Path, Concave Monge Property, Farthest-point Voronoi Diagram, Delaunay Triangulation, Dynamic Programming
	\end{abstract}

	
	\section{Introduction}
	\label{sec:intro}
 Given a set $\cal B$ of blue points and a set $\cal R$ of red points above a horizontal line $\ell$, each point $p_i \in {\cal B} \cup {\cal R}$ has a weight $w_i > 0$ if $p_i \in {\cal B}$ and $w_i < 0$ if $p_i \in R$. Let $|{\cal B} \cup {\cal R}|=n$, and let $d^0$ denote the interior of any geometric object $d$, i.e., $d^0=d\backslash {\partial d}$ where $\partial d$ denotes the boundary of $d$. The objective is to pack $k$ non-overlapping congruent disks $d_1, d_2, \dots, d_k$ of minimum radius, centered on $\ell$, such that $\sum\limits_{j=1}^{k}\sum\limits_{i\in [n]: \exists p_i\in {\cal  R}, p_i\in d^0_j} w_i + \sum\limits_{j=1}^{k}\sum\limits_{i\in [n]: \exists p_i\in {\cal  B}, p_i\in d_j} w_i$ is maximized, i.e., the sum of the weights of the points covered by $\bigcup_{j=1}^{k}d_j$ is maximized. We name this problem a Constrained Semi-Obnoxious Facility Location (\textsc{CSofl}) problem on a Line. 
 
 Typically, facility location problems involve two types of facilities: desirable ones like hospitals, fire stations, and post offices that should be located as close as possible to demand points (population centers), and undesirable ones like chemical factories, nuclear plants, and dumping yards that should be located as far away as possible from demand points to minimize their negative impact. However, the semi-obnoxious facility location \textsc{(Sofl)} problems have the unified objective of optimizing both negative and positive impacts on the repelling and attractive demand sites, respectively. In \textsc{Sofl} problems, the aim is to locate facilities at an optimal distance from both attractive and repulsive demand points. This creates a bi-objective problem where two objectives must be balanced. For example, when building an airport, it should be located far enough from the city to avoid noise pollution but close enough to customers to minimize transportation costs.

 In \cite{ma1994}, the semi-obnoxious facility location problem is modeled as Weber's problem, where the repulsive points are assigned with negative weights. This problem is solved by designing a branch and bound method with the help of rectangular subdivisions \cite{ma1994}. 
The problem of locating multiple capacitated semi-obnoxious facility location problem is solved using a bi-objective evolutionary strategy algorithm where the objective is to minimize both non-social and social costs \cite{te2019}.
The problem of locating a single semi-obnoxious facility within a bounded region is studied by constructing efficient sets (the endpoints of the efficient segments) \cite{me2003}.
A bi-objective mixed integer linear programming formulation was introduced and applied to this semi-obnoxious facility location problem \cite{co2013}.

Golpayegani et al. \cite{go2014} introduced a semi-obnoxious median line problem in the plane using Euclidean norm and proposed a particle swarm optimization algorithm. Later, Golpayegani et al. \cite{go2017} proposed a particle swarm optimization that solved the rectilinear case of the semi-obnoxious median line problem. Recently, Gholami and Fathali \cite{gh2021} solved the circular semi-obnoxious facility location problem in the Euclidean plane using a cuckoo optimization algorithm which is known to be a meta-heuristic method.
The problem of locating a single semi-obnoxious facility within a bounded region is studied by constructing efficient sets.
Wagner \cite{wa2015} gave duality results in his thesis for a non-convex single semi-obnoxious facility location problem in the Euclidean space. Singireddy and Basappa \cite{vi2021,sing2022} studied the $k$ obnoxious facility location problem restricted to a line segment. They initially proposed an $(1-\epsilon)$-approximation algorithm \cite{vi2021}, and then two exact algorithms based on two different approaches that run in $O((nk)^2)$ and $O((n+k)^2)$ time, respectively, for any $k>0$  and finally, an $O(n\log{n})$ time algorithm for $k=2$ \cite{sing2022}. In \cite{vi2021}, they also  examined the weighted variant of the problem (where the influencing range of obnoxious facilities is fixed and demand points are weighted). They gave a dynamic programming-based solution that runs in $O(n^3k)$ time for this variant. Subsequently, Zhang \cite{zhang2022} refined this result to $O(nk\alpha(nk)\log^3{nk})$ by reducing the problem to the $k$-link shortest path problem on a complete, weighted directed acyclic graph whose edge weights satisfy the convex Monge property, where $\alpha(\cdot)$ refers to the inverse Ackermann function. Section \ref{sec4} of this paper follows  a similar strategy of reducing the weighted \textsc{CSofl} to the $k$-link path problem, but the edge weights satisfy the concave Monge property. 

The \textsc{Sofl} problem is also closer to the class of geometric separability problems, where we need to separate the given two sets of points with a linear or non-linear boundary or surface in a high-dimensional space. Geometric separability is an important concept in machine learning and pattern recognition. It is used to determine whether a set of data can be classified into distinct categories (say, good and bad points type) using a specific algorithm or model. Then, given an arbitrary data point, we can predict whether this point is good or bad depending on which side of the separation boundary hyperplane it falls. O'Rourke et al. \cite{o1986} gave a linear time algorithm based on linear programming to check whether a circular separation exists. They also showed that the smallest disk and the largest separating circle could be found in $O(n)$ and $O(n\log {n})$ time, respectively. If a convex polygon with $k$ sides separation exists, then for $k=\Theta(n)$, the lower bound for computing the minimum enclosing convex polygon with $k$ sides is $\Omega(n\log{n})$ \cite{E1988} and can be solved in $O(nk)$ time. While the separability problem using a simple polygon \cite{fek1992} was shown to be {\tt NP-hard}, Mitchell \cite{mit1993} gave $(\log{n})$-approximation algorithm for an arbitrary simple polygon. Recently, Abidha and Ashok \cite{abidha2022} have explored the geometric separability problems by examining rectangular annuli with fixed (the axis-aligned) and arbitrary orientation, square annuli with a fixed orientation, and an orthogonal convex polygon. For rectangular annuli with a fixed orientation, they gave $O(n\log {n})$ time algorithm. They gave $O(n^2\log {n})$ time algorithm for cases with arbitrary orientation. For a fixed square case, the running time of their algorithm is $O(n\log^2{n})$, while for the orthogonal convex polygonal cases, it is $O(n\log{n})$ time.
	
\section{Preliminaries}
This section briefly introduces various notations and definitions that will be used in further sections.   

Let ${\cal L}_{\textsc{can}}$ denote the set of candidate radii and $r_{\textsc{can}}\in {\cal L}_{\textsc{can}}$ a candidate radius. The optimal radius is denoted as $r_{opt}$. Let $dist(u,v)$ denote the minimum Euclidean distance between two points $u$ and $v$.
Given a graph $G(V,E)$, the weight of an edge is denoted as $w(i,j)$ where $\overline{ij}\in E$. The path between any two vertices $v,u\in V$ is denoted as $\Pi(v,u)$.


\begin{definition}{\textbf{Configurations}:}
Arrangement of $k$-disks in any feasible solution  to the \textsc{CSofl} problem with different placements of red and blue points on the boundaries of the disks (critical regions). A configuration is said to be critical if it corresponds to some candidate radius $r_{\textsc{can}}\in {\cal L}_{\textsc{can}}$.
\end{definition}

\begin{definition}\textbf{DAG:}
   It stands for Directed Acyclic Graph. It is a directed graph that has no directed cycles. In other words, it is a graph consisting of a set of nodes connected by directed edges, where the edges have a specific direction. There is no way to start at any node and follow a sequence of edges that eventually loops back to that node. 
\end{definition}

\begin{definition} \textbf{The minimum weight $k$-link path:}
    Given a complete weighted DAG $G(V,E)$, and two vertices $s$ (source) and $t$ (target), the minimum weight $k$-link path problem seeks to find a minimum weight path from $s$ to $t$ such that the path has exactly $k$ edges (links) in it.
\end{definition}

\begin{definition}
    \textbf{Concave Monge property:} The weight function $w$ for a given weighted, complete DAG $G(V,E)$ satisfies the concave Monge property if for all $i,j\in V$ we have the inequality $w(i,j)+w(i+1,j+1)\leq w(i,j+1)+w(i+1,j)$ satisfied, where $1<i+1<j<n$.
\end{definition}

\noindent The outline of the algorithm for the \textsc{CSofl} problem is as follows:
\begin{enumerate}
\item First, all possible configurations of the $k$ disks and red and blue points in any feasible solution to the \textsc{CSofl} problem are identified. We show that a finite number of distinct configurations exist, and there are specifically $O(1)$ distinct critical-configuration types. 
\item The next step entails computing all possible candidate radii, ${\cal L}_{\textsc{can}}$, where we have one ${r}_{\textsc{can}}$ corresponding to each of the configurations identified in the previous step.
\item After obtaining a candidate radius $r_{\textsc{can}}$, the given instance of the \textsc{CSofl} problem will be transformed into an instance of the problem of computing a minimum weight $k$-link path problem on a complete weighted DAG $G$.
\item The semi-obnoxious facilities (disks) should then be positioned (i.e., the centers of these disks are to be positioned) at the points on $\ell$ corresponding to vertices of the aforementioned minimum weight $k$-link path $\Pi_k^*(s,t)$ in $G$. The total weight of the points covered by these facilities can be computed using the $\Pi_k^*(s,t)$ weight.
\item To determine the set of all radii ${\cal L_{\textsc{can}}}=\{\lambda_1,\lambda_2,\dots \}$ for which the total weight of the covered points is the largest, the above process must be repeated for every candidate radius $r_{\textsc{can}}$.
\item Finally, the locations of the $k$ semi-obnoxious facilities placed with the smallest $\lambda\in {\cal L}_{\textsc{can}}$ and covering the points with the largest total weight is returned as the output.
\end{enumerate}

\section{Computing the candidate radii}\label{sec3}
In this section, we find all the candidate radii by considering all configurations involving disks as well as blue and red points.

\begin{itemize}
\setlength{\itemindent}{0.92 in}
 \item [\textbf{\textit{Configuration}-0:}] Suppose that all the red points lie closer to $\ell$ than the blue points and have significantly more negative weights than the blue points (see Figure \ref{fig:oo}). In this scenario, covering any blue points would also cover some red points since the disks must be centered on $\ell$. This, in turn, would result in a negative total weight. As a result, we can opt to keep zero-radius disks that do not cover any points rather than covering any of the blue points. This way, the maximum weight will be zero. It also implies the following observation.

 \begin{observation}
  An optimal (feasible) solution always exists for any given problem instance.
 \end{observation}
 
 \begin{figure}[!htb]
    \centering
    \includegraphics[width=.7\textwidth]{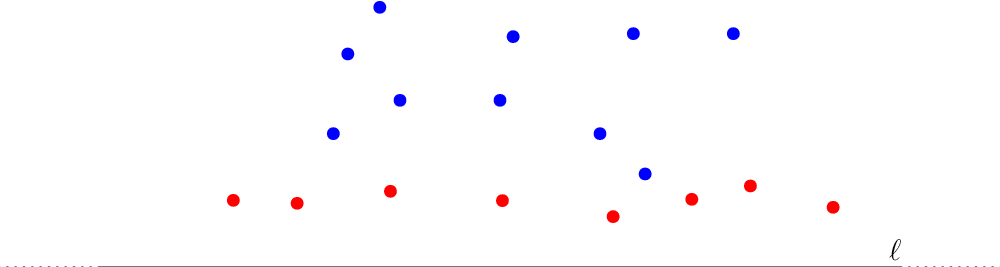}
    \caption{Configuration-0.}
    \label{fig:oo}
\end{figure}

 \item [\textbf{\textit{Configuration}-1:}] We consider a specific configuration in which the radius of the disks in the optimal solution is determined by only one blue point. As shown in Figure \ref{fig:op}, we can observe that the disks $d_i$ and $d_{i+1}$, which have one blue point on each of their boundaries, will have a smaller (optimal) radius compared to the dotted disk $d_i'$, which also covers the same blue points. This is because our problem is to find the minimum radius disks that cover the maximum weight.
  \begin{figure}[!htb]
    \centering
    \includegraphics[width=.7\textwidth]{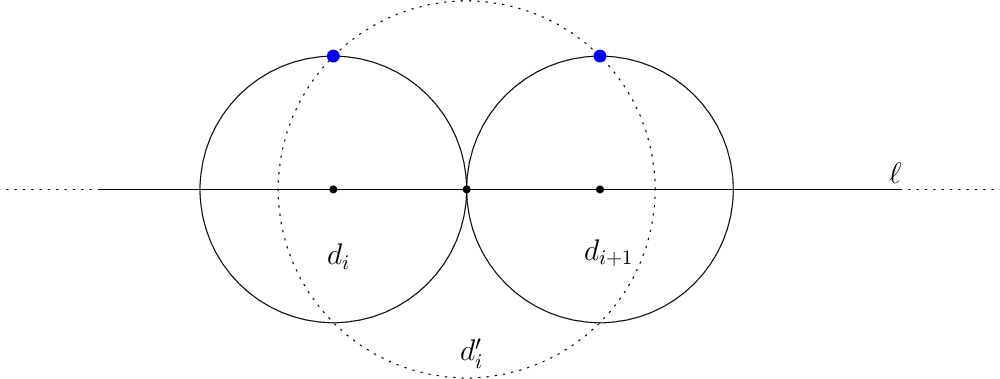}
    \caption{Configuration-1.}
    \label{fig:op}
\end{figure}
Here, we consider at least one of the blue points lying on the boundary of at least one of the $k$ disks, which determines the radius of the disks in the optimal solution for a radius greater than zero. Observe that the radius of the disk is the $y$-coordinate value of the point that lies on the disk boundary. Hence, we add $O(n)$ candidate radii to ${\cal L}_{\textsc{can}}$ and the radii are $r_{\textsc{can}}=y_{p_i}$, where $y_{p_i}$ denotes the $y$-coordinate of the point $p_i$ for each $p_i\in {\cal B}$, $i=1,2,\dots,n$. 
 
 
 \item [\textbf{\textit{Configuration}-2:}] In this scenario, we consider the case where the optimal solution is determined by two points on the boundary of at least one of the $k$ disks, which can either be two blue points or one blue and one red point. We notice that no two red points on any disk's boundary will determine the disk's radius, as we can further reduce the disk's radius until its boundary touches at least one of the blue points. To calculate the candidate radii for the disks, we proceed as follows:
 

 

Consider a point $p_i\in {\cal B}$ and a point $p_j\in {\cal R}$. If they determine the minimum radius of the disks in a solution to \textsc{CSofl} problem, then the candidate radius $r_{\textsc{can}}$ can be computed by drawing a bisector line $\ell_{i,j}$ until $\ell_{i,j}$ cuts across $\ell$ where $p_i$ is in the counter-clockwise direction from $\ell_{i,j}$ and $p_j$ in the clockwise direction from $\ell_{i,j}$ (see Figure \ref{fig:lij}).

\begin{figure}[!htb]
    \centering
    \includegraphics[width=.7\textwidth]{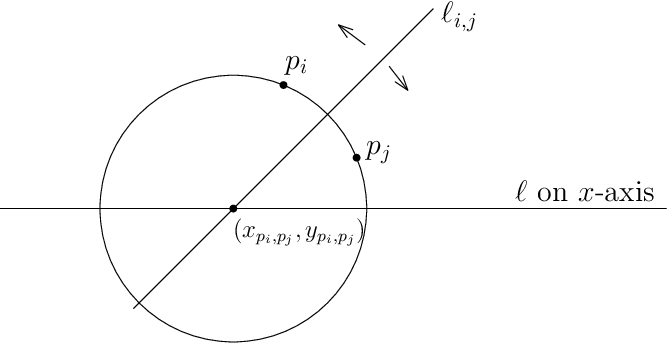}
    \caption{Illustration of calculating $r_{\textsc{can}}$.}
    \label{fig:lij}
\end{figure}

Let $(x_{p_i,p_j},y_{p_i,p_j})$ be the center of the disk, $(x_{p_i},y_{p_i})$ and  $(x_{p_j},y_{p_j})$ are the coordinates of the points $p_i$ and $p_j$ respectively. Then, we have \[(x_{p_i}-x_{p_i,p_j})^2+(y_{p_i}-y_{p_i,p_j})^2=(x_{p_j}-x_{p_i,p_j})^2+(y_{p_j}-y_{p_i,p_j})^2\]
After simplification, we have $r_{\textsc{can}}=\sqrt{(x_{p_i}-x_{p_i,p_j})^2+y_{p_i}^2}$ for the two cases:
\begin{itemize}
    \item $p_i\in {\cal B}$ and $p_j\in {\cal R}$.
    \item $p_i,p_j\in {\cal B}$.
\end{itemize}
where $x_{p_i,p_j}=\frac{(y_{p_j}-y_{p_i})(y_{p_j}+y_{p_i})}{2(x_{p_j}-x_{p_i})}+\frac{(x_{p_j}+x_{p_i})}{2}$, $y_{p_i,p_j}=0$. Here, as we consider a candidate radius for every pair of points except the red-red pair, we add $O(n^2)$ candidate radii to ${\cal L}_{\textsc{can}}$ for this critical-configuration type.
\end{itemize}
\begin{observation} \label{ob2}
    There are only $O(1)$ critical configuration types.
\end{observation}
\begin{proof}
    In any of the optimal placements of the disks, either one blue point or a pair of blue points, or a pair of red and blue points will determine the disk's radius. Even though there may be more than two points lying on the boundary of the disks in any optimal packing, only two points among them will determine the radius of the disk since there exists a unique disk that passes through two points and is centered on $\ell$. Hence, any configurations of points lying on the boundary of the disk can be transformed into any of the above-mentioned configurations by perturbing the center or reducing the radius of the disks. Therefore, we have a constant number of critical configuration types, namely, four types, including the configuration-0.
\end{proof}
\begin{lemma} \label{lemma7}
    $|{\cal L}_{\textsc{can}}|=O(n^2)$. 
\end{lemma}
\begin{proof}
    It follows from Observation \ref{ob2} since a constant number of critical configuration types (viz. no points, one blue point, a pair of blue points, and a pair of blue and red points) contribute to the candidate radii. In any of the configurations, at most, two points will determine the radius of the disks. Hence we have $O(n^2)$ candidate radii corresponding to that configuration. Thus, the lemma follows.
\end{proof}
\section{Transformation to the minimum weight $k$-link path problem}\label{sec4}
In this section, we demonstrate that the \textsc{CSofl} problem can be reduced to the problem of computing a minimum weight $k$-link path between a pair of vertices in a weighted DAG $G(V,E)$. Each edge $\overline{ij}\in E$ in $G$ is assigned a weight $w: (i,j)\rightarrow \mathbb{R}$ that is either a positive or negative real number $w(i,j)\in \mathbb{R}$.

The minimum weight $k$-link path $\Pi(s,t)$ is a path from the source $s$ to the target vertex $t$, consisting of exactly $k$ edges, and has the minimum total weight among all $k$-link paths between $s$ and $t$, where the weight of a $k$-link path is the sum of weights of the edges in the path, i.e., \[w(\Pi(s\rightarrow i_1 \rightarrow i_2 \rightarrow \dots \rightarrow i_{k-1}\rightarrow t))=\sum\limits_{j=1}^{k-2}w(i_j,i_{j+1})+w(s,i_1)+w(i_{k-1},t)\]

Let $\lambda=r_{\textsc{can}}$. Next, we transform an instance of the \textsc{CSofl} problem to an instance of $k$-link path problem on a DAG $G(V,E)$ as follows:

Let us call the maximal interval $f_i^+=[l_i,r_i]$ on $\ell$ as the influence interval (within which a facility or a disk with radius $r_{\textsc{can}}$ centered will influence or cover the point $p_i$) for the point $p_i\in {\cal B}$ if the distance between any point on $[l_i,r_i]$ and $p_i$ is at most $r_{\textsc{can}}$. Similarly, $f_i^-=[l_i,r_i]$ is the influence interval on $\ell$ for $p_i\in {\cal R}$.

Let the set of all influence intervals be $F=\{f_i^+ \mid i\in [n], p_i\in {\cal B}\}\cup \{f_i^- \mid i\in [n], p_i\in {\cal R}\}$. Let the vertex set $V=\{l_1,r_1,l_2,r_2,\dots,l_n,r_n\}$ be the end points of the intervals in $F$. For each $l_i, i\in [n]$, we also add $2(k-1)$ extra vertices corresponding to points on $\ell$ at distance $l_i+2\lambda,l_i+4\lambda,\dots,l_i+2(k-1)\lambda, l_i-2\lambda, l_i-4\lambda,\dots,l_i-2(k-1) \lambda$, to $V$. Similarly, we add $2(k-1)$ vertices for every $r_i, i\in [n]$, placed at points at distance $r_i-2\lambda,r_i-4\lambda,\dots,r_i-2(k-1)\lambda, r_i+2 \lambda, r_i+4\lambda,\dots,r_i-2(k-1) \lambda$.
The addition of these extra $2(k-1)$ points on the sides of both endpoints of each influence interval is because it is possible to have disks centered at points on $\ell$ other than the endpoints of influence intervals in an optimal solution (see Figure \ref{fig:kmo} for an illustration). However, at least one disk must be centered at an endpoint in any optimal solution. In Figure \ref{fig:kmo}, we can observe that the disks $d_{i-1}$ and $d_{i+1}$ are not centered at any endpoint of the intervals in $F$ since none of the points in ${\cal B}\cup {\cal R}$ lie on their boundary.
\begin{figure}[!htb]
    \centering
    \includegraphics[width=\textwidth]{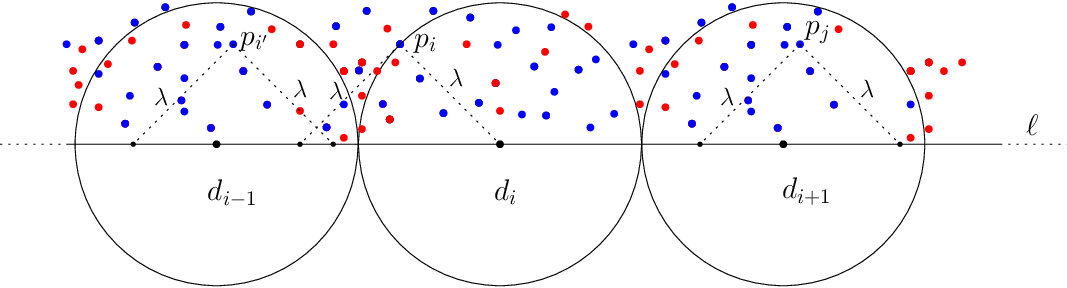}
    \caption{An optimal packing of $3$ disks for a candidate radius $\lambda$, $d_i$ is centered at an endpoint of the influence interval due to $p_i$.}
    \label{fig:kmo}
\end{figure}

Without loss of generality, the vertices may be relabeled as $V=\{v_1,v_2,\dots,v_m\}$ based on the increasing order of $x$-coordinates of all $l_i$ and $r_i$ for $i\in [n]$, and the extra added points, where $m=O(kn)$. Furthermore, we can update $V$ so that all the corresponding points in $V$ have distinct $x$-coordinates. Let $s$ and $t$ be the points placed on $\ell$ at a distance of $2k\lambda$ from the left endpoint of the leftmost interval $l_1$ and from the right endpoint of the rightmost interval $r_{\textsc{l}}$, respectively, where $[l_{\textsc{l}},r_{\textsc{l}}]$ denotes the rightmost influence interval (see Figure \ref{fig:st}). Note that $s$ lies on the left of all the points in $V$, and $t$ lies on the right of all the points in $V$.

\begin{figure}[!htb]
    \centering
    \includegraphics[width=\textwidth]{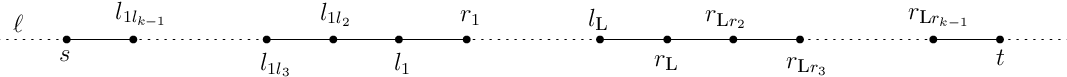}
    \caption{Adding of extra-points including $s$ and $t$.}
    \label{fig:st}
\end{figure}

\begin{figure}[!htb]
    \centering
    \includegraphics[width=0.8\textwidth]{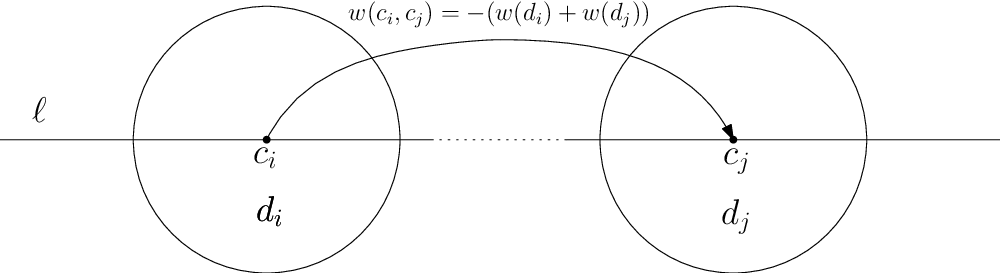}
    \caption{Calculation of the weight of an edge.}
    \label{fig:ewt}
\end{figure}

Let $w(d_i)$ denote the total weight of the points covered by the disk with radius $\lambda$ centered at $c_i\in V$. We calculate $w(d_i)$ for the disks centered at each point $c_i\in V$.

\noindent Now, the  weight of an edge $\overline{ij}\in E$ is calculated as follows:
\begin{itemize}
    \item $w(i,j)= +\infty$ if the $dist(i,j)<2\lambda$.
    \item $w(i,j)=-(w(d_i)+w(d_j))$ if $dist(i,j)\geq 2\lambda$ for all $i,j\in V$, i.e., we add a directed edge between every pair of vertices $i,j\in V$ ($i<j$), if the distance between them is at least $2\lambda$ and we add the corresponding weights (see Figure \ref{fig:ewt}) and negate it. 
    \item Clearly observe that $w(d_s)$ and $w(d_t)$ is zero. It is also zero for the disks centered at first (at most) $k-1$ and last (at most) $k-1$ points (since they are the points on $\ell$ which are separated by a distance of at least $2\lambda$ on the left of $l_1$ and the right of $r_{\textsc{l}}$), respectively for every endpoint of the influence interval.   
\end{itemize}

 Without loss of generality, let $G'(V',E')$ be the graph obtained by the above transformation. There will be $O(nk)$ vertices in $V'$, and a directed edge from $i$ to $j$ for all $i,j\in V'$ such that $i<j$. Then, $G'$ is a complete DAG with $|V'|=O(nk)$ vertices and $|E'|=O(n^2k^2)$ edges, and every edge $(i,j)\in E'$ is assigned a weight as discussed above. Hence, we have the following lemma.
 \begin{lemma} \label{lemma8}
     The graph $G'$ can be constructed in $O(n^2k^2)$ time.
 \end{lemma} 
 \begin{proof}
    We start by considering the way we constructed $G'$. Every demand point $p_i\in {\cal B}\cup {\cal R}$ can contribute at most two endpoints of an interval on $\ell$ at a distance of $2\lambda$ from each other. If we center a disk on that interval, the demand points will either lie on the boundary or the interior of the disk. Next, we add $2(k-1)$ points on both sides of each endpoint on $\ell$ with a separation distance of $2\lambda$ between any two consecutive of them. Thus, we have a total of $O(nk)$ points on $\ell$, which includes $s$ and $t$ and are added to $V'$. Now, from every point in $V'$, we add a weighted directed edge to all the points of $V'$ that lie on the right of that point on $\ell$. This will result in a total of $O(n^2k^2)$ edges, where each edge is assigned a corresponding weight, as discussed earlier. Therefore, the resulting graph $G'(V',E')$ has $O(nk)$ vertices, $O(n^2k^2)$ edges, and can be constructed in $O(n^2k^2)$ time.
 \end{proof}
The edge weights of $G'$ will satisfy the concave Monge property for any four vertices $i, i+1, j, j+1$ of $G'$ such that $i < i+1 < j < j+1$, the weights of the directed edges from $i$ to $j$ and from $i+1$ to $j+1$ are not greater than the weights of the directed edges from $i$ to $j+1$ and from $i+1$ to $j$. 
\begin{observation}
    The edge weights of $G'$ satisfy the concave Monge property.
\end{observation}
\begin{proof}
\begin{figure}[!htb]
    \centering
    \includegraphics[width=0.8\textwidth]{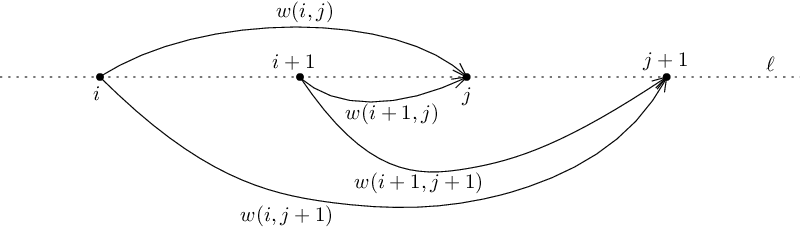}
    \caption{Any four vertices of $G'$ that satisfy the concave Monge property.}
    \label{fig:mng}
\end{figure}

    Recall the definition of the concave Monge property, i.e., $w(i,j)+w(i+1,j+1)\leq w(i,j+1)+w(i+1,j)$ (see Figure \ref{fig:mng}). The weights assigned to the edges of $G'$ are assigned based on the following two rules:
    \begin{enumerate}
     \item $w(i,j)= +\infty$ if the $dist(i,j)<2\lambda$.
     \item $w(i,j)=-(w(d_i)+w(d_j))$ if $dist(i,j)\geq 2\lambda$ for all $i,j\in V$
    \end{enumerate}
    Suppose we select any four vertices with their index labels satisfying $i<i+1<j<j+1$ in $G'$ such that the distance between them is at least $2\lambda$. Then, according to rule 2, the corresponding weights assigned to the edges satisfy $w(i,j)+w(i+1,j+1)= w(i,j+1)+w(i+1,j)$. Now, consider another set of four vertices $i<i+1<j<j+1$ in $G'$ such that the distance between the two closest points among them is less than $2\lambda$, i.e., $dist(i+1,j)<2\lambda$. According to rule 1, the corresponding weight assigned to the edges between these vertices is $+\infty$. Then, we have $w(i,j)+w(i+1,j+1)< w(i,j+1)+w(i+1,j)$, which satisfies the concave Monge property. Finally, suppose all four selected vertices $i<i+1<j<j+1$ in $G'$ have a distance between any two consecutive of them less than $2\lambda$. In this case, according to rule 1, the weights assigned to the corresponding edges satisfy $w(i,j)+w(i+1,j+1)= w(i,j+1)+w(i+1,j)$.

Therefore, we have shown that if we select any four vertices $i<i+1<j<j+1$ in $G'$, the weights of all these edges satisfy the concave Monge property. Thus, the observation follows.
\end{proof}

Since the constructed graph $G'$ is a weighted complete DAG and its edge weights satisfy the concave Monge property, we have the following theorem for finding the minimum weight $(k+1)$-link path between any pair of vertices of $G'$.
  \begin{theorem} \cite{agga1993}\label{thm10}
    The minimum weight $(k+1)$-link path $\Pi_{(k+1)}^*(s,t)$ in $G'$ can be computed in $O(nk\sqrt{k\log{(nk)}})$ time.
\end{theorem}

\begin{theorem}\label{thm11}
We can solve the \textsc{CSofl} problem in polynomial time.
\end{theorem}
\begin{proof}
It follows from Lemma \ref{lemma7}, Lemma \ref{lemma8} and Theorem \ref{thm10}. The running time of the algorithm is $n^2\cdot(O(n^2k^2)+O(nk\sqrt{k\log{(nk)}}))=O(n^4k^2)$.
The algorithm will return the minimum $r_{\textsc{can}}=r_{opt}$, corresponding to which the computed $(k+1)$-link path between $s$ and $t$ has the minimum total weight $w(\Pi_{(k+1)}^*(s,t))$. The disks  with radius $r_{opt}$ can be centered at $k$ internal vertices of the $(k+1)$-link path (excluding the terminal vertices $s$ and $t$). The total weight is $\sum_{i\in Sol}w(d_i)= -w(\Pi_{(k+1)}^*(s,t))/2$, where $d_i$ is a disk centered at $i$ having radius $r_{opt}$ and $Sol=V(\Pi_{(k+1)}^*(s,t))\setminus\{s, t\}$ is the set of vertices (the corresponding points on $\ell$) of $(k+1)$-link path except $s$ and $t$. 
\end{proof}
\noindent \textbf{Improvement:} Recall that the points corresponding to the vertices in $V'$ are labeled as $1, 2, \ldots, m$, where $m=O(nk)$. Here, we show that we can improve the runtime of Theorem \ref{thm11} (by almost linear factor) by not explicitly constructing the complete graph $G'$. As we have seen in the proof of Lemma \ref{lemma8}, this construction requires $O(n^2k^2)$ time for each of $O(n^2)$ candidate radius. However, for every candidate radius $r_{\textsc{can}}$, we need to precompute the two arrays $w[ ]$ and $p[ ]$, each of size $O(nk)$. Here, $w[i]$ stores the sum of weights of the demand points covered by the disk of radius $r_{\textsc{can}}$ centered at a point labeled $i\in V'$ on $\ell$, for each $i\in [m]$. The element $p[i]$ stores the index $i'\in V'$  of the rightmost point at a distance of at least $2\lambda$ from $i$ such that $i'<i$. We now give a dynamic program algorithm to compute the maximum weight of a $(k+1)$-link path from point $1$ to point $m$ (here, the points labeled $1$ and $m$ are the vertices $s$ and $t$ respectively, in $G'$). For a pair of points $i, j$ ($ i < j$) on $\ell$; we redefine the weight of the edge $(i,j)$ to be equal to $w(j)$ as we need not negate the sum of weights and construct the whole directed graph .

We define the subproblem $\phi(i,j)$ as the problem of finding the maximum weight $j$-link path in the subgraph $G'_i$ induced by the vertices $1, 2, \dots, i$. That is,  $\phi(i,j)=\max\limits_{i'}\{w(\Pi_j(1,i'))\}$, where $(j+1)\leq i'\leq i$. Then we have the following recurrence:
\begin{equation}\label{dpeq}
    \phi(i,j)=\max\{\phi(i-1,j),\phi(p[i],j-1)+w[i]\}
\end{equation}

 As $i=1,2,\dots, O(nk)$, and $j=1,2,\dots,k+1$, there are $nk\cdot k=O(nk^2)$ entries in the DP table table $\phi(i,j)$, each requiring $O(1)$ time to compute.
 
Hence, for a given $r_{\textsc{can}}$, the bottom-up implementation of the above dynamic programming algorithm takes time $O(nk^2)$ provided that we have the entries $w[i]$ and $p[i]$ precomputed for each $i\in[m]$.
\begin{theorem}
   The above-improved algorithm for the \textsc{CSofl} problem has the time complexity of $O(n^3k \cdot\max{(n, k)})$.
\end{theorem}
\begin{proof} The proof is as follows:
    \begin{itemize}
        \item There are $O(n^2)$ candidate radii in ${\cal L}_{\textsc{can}}$.
        \item For each candidate radius $\lambda\in {\cal L}_{\textsc{can}}$, we find $O(nk)$ points on $\ell$ as discussed above.
        \item To compute the weight $w[i]$ of each point $i\in V'$ on $\ell$, we answer a circular range reporting query \cite{chan2016}, which takes $O(\log n+\kappa)$ time for a query circle centered at each of $O(nk)$ points on $\ell$, where $\kappa$ is the number of points reported. It has a preprocessing time of $O(n\log{n})$ and requires $O(n)$ space.
        \item For each $\lambda\in {\cal L}_{\textsc{can}}$, the above dynamic programming algorithm will take $O(nk^2)$ time. However, the bottleneck in computing the optimal value $\phi(m,k+1)$ is in computing the arrays $w[ ]$ for each of $O(n^2)$ candidate radius. The total time for computing this array is $n\log{n}+\sum\limits_{j=1}^{|{\cal L}_{\textsc{can}}|}(\sum\limits_{i=1}^{m}(\log{n}+\kappa_i))$, where $\kappa_i$ is the number of points reported by the query algorithm for a given query disk centered at $i\in[m]$, where $m=O(nk)$. We see that $\sum\limits_{j=1}^{|{\cal L}_{\textsc{can}}|}(\sum\limits_{i=1}^{m}(\log{n}+\kappa_i))=n^3k\log{n}+\sum\limits_{j=1}^{|{\cal L}_{\textsc{can}}|}n\chi$, where $\chi=\max\{\chi_1, \chi_2, \ldots, \chi_{O(n^2)}\}$, and $\chi_j$ is the ply\footnote{The ply of a set $P$ of points with respect to a set $D$ of disks $d_1, d_2, \ldots,$ is the largest number of those disks in $\bigcup_{i=1,2.\ldots} d_i$ whose intersection contains a point $p\in P$.} of the pointset ${\cal B}\bigcup{\cal R}$ with respect to the set of all  disks $i\in [m]$ for a candidate radius $\lambda_j \in {\cal L}_{\textsc{can}}$. Observe that the ply of ${\cal B}\bigcup{\cal R}$ for a given set of $O(nk)$ disks is $O(n)$ only since a point from ${\cal B}\bigcup{\cal R}$ lying in a disk centered at an endpoint $i$ of influence interval can not be contained inside the $2(k-1)$ disks centered at distances $l_i+2\lambda,l_i+4\lambda,\dots,l_i+2(k-1)\lambda, l_i-2\lambda, l_i-4\lambda,\dots,l_i-2(k-1) \lambda$. Therefore, $\sum\limits_{j=1}^{|{\cal L}_{\textsc{can}}|}(\sum\limits_{i=1}^{m}\kappa_i)\leq\sum\limits_{j=1}^{|{\cal L}_{\textsc{can}}|}(nk\chi)=O(n^4k)$ as $\chi=O(n)$.
        \item Hence the total running time is $n^2\cdot(O(nk^2)+O(nk\log{n})+O(n^2k))=O(n^3k\cdot \max{(n,k)})$.
    \end{itemize}
Now, we prove the correctness of the dynamic programming recurrence relation \ref{dpeq} by inducting on the number of disks placed.

\noindent \textbf{Correctness:} 
Fix a radius $\lambda\in {\cal L}_{\textsc{can}}$.
By induction on $i+j$, we can prove the recurrence relation \ref{dpeq} is correct, as follows. For the base case $j=1$, $i\geq 2$, we have $\phi(i,1) = w(\Pi_1(1,i))=\max\limits_{2\leq q\leq i}{(w[q])}$, which is the maximum weight of a 1-link path originating at vertex 1 in the subgraph $G'_i$ induced by the vertices $1,2,\dots,i$. This is the optimal solution for the subproblem $\phi(i,1)$. For the base case $j\geq 2$, $i=2$, we have that $\phi(i,j)=\phi(i,2)$ as the weight contributed by the remaining $j-2$ disks centered on $\ell$ is zero.

Let $p[s]= \max\{q\mid (dist(s,q)\geq 2\lambda, 2 \leq q<s\}$ for $2\leq s \leq i$. Assume that the recurrence relation holds for all subproblems $\phi(i',j')$, where $(i'+j')<(i+j)$. Consider the subproblem $\phi(i,j)$, and for solving this subproblem, we consider two cases for the vertex $i$: either $\Pi_j(1,i)$ uses vertex $i$ or it doesn't.

\noindent Case $1$: $\Pi_j(1,i)$ does not use vertex $i$. In this case, $\Pi_j(1,i)$ is also an optimal $j$-link path in $G'_{i-1}$ by induction hypothesis. Therefore, the optimal solution for the subproblem $\phi(i,j)$ is the same as for the subproblem $\phi(i-1,j)$.

\noindent Case $2$: $\Pi_j(1,i)$ uses vertex $i$. Let $i' = p[i]$, which is the predecessor of $i$ on $\Pi_j(1,i)$. Then, $\Pi_j(1,i)$ can be decomposed into two parts: an optimal $(j-1)$-link path in $G'_{i'}$ (by induction hypothesis), denoted by $\Pi_{j-1}(1,i')$, and the edge $(i',i)$ with weight $w[i]$. Since $\Pi_j(1,i)$ is an optimal path ending at $i$ in the subgraph $G_i'$, the weight of $\Pi_j(1,i)$ is equal to the sum of the weights of $\Pi_{j-1}(1,i')$ and $(i',i)$, i.e., $w(\Pi_j(1,i)) = w(\Pi_{j-1}(1,i')) + w[i]$.

Therefore, the optimal solution for the subproblem $\phi(i,j)$ is the maximum weight of all $j$-link paths in $G'_i$. This is achieved by either taking the optimal solution for the subproblem $\phi(i-1,j)$ or by taking the optimal solution for the subproblem $\phi(i',j-1)$ and adding the weight of edge $(i',i)$, i.e., $\phi(i,j) = \max{(\phi(i-1,j),\phi(p[i],j-1)+w[i])}$.

By using the recurrence relation for $\phi(i-1,j)$ and $\phi(p[i],j-1)$, which can be computed by solving the subproblems $\phi(i-1,j-1)$ and $\phi(p[i],j-1)$. Therefore, we can use the recurrence relation to obtain the optimal solution for $\phi(i,j)$.

By the principle of mathematical induction, the recurrence relation holds for all subproblems $\phi(i,j)$, where $1\leq j \leq k$. Given that we have precomputed all the values in the arrays $w[]$ and $p[]$, it takes constant time to compute the optimal solution to each subproblem by combining optimal solutions to smaller subproblems. Further, we have $O(nk^2)$ distinct subproblems in total for the recurrence. Hence, the overall time complexity of the algorithm is $O(nk^2)$.
\end{proof}

\section{Special cases of \textsc{CSofl}}
In this section, we consider the following two special cases of the \textsc{CSofl} problem with some specific application.
\begin{problem}
\textbf{\textsc{AllBlue-MinRed: }}The problem aims to cover all blue points while covering the minimum number of red points. To solve this problem, we modify the weights of the demand points as follows: for every point $p_i\in {\cal R}$, let $w_i=\delta$ and for every point $p_i\in {\cal B}$, the weight $w_i>-|{\cal R}|\delta$, where $\delta\in \mathbb{R}$ is an arbitrary real value and $\delta<0$. 
\end{problem}

This problem has some specific applications in defense, as will be discussed: assuming a scenario where there are two groups of points along a horizontal line, one represented by blue points (enemy forces) and the other by red points (civilians), the goal is to determine the center locations and blast radius required for a fixed number of explosives to target all enemy forces while isolating the civilians as much as possible. Alternatively, suppose the scenario is such that the red points represent enemy forces, and the objective is to establish wireless communication among our own forces (represented by blue points). In that case, the goal is to place $k$ base stations to cover all blue forces while minimizing the transmissions intercepted by the red forces (enemy forces).
\begin{problem}
\textbf{\textsc{MaxBlue-NoRed: }}In this problem, we need to cover the maximum number of blue points, and at the same time, none of the red points need to be covered. To solve this problem, we modify the weights of the demand points as follows: for every point $p_i\in {\cal B}$, let $w_i=\delta$, and for every point $p_i\in {\cal R}$, the weight $w_i<-|{\cal B}|\delta$, where $\delta\in \mathbb{R}$ is an arbitrary real value and $\delta>0$. 
\end{problem}
This problem has the following specific applications: place a set of $k$ sensors on a horizontal line to cover as many blue points as possible while avoiding red ones. This scenario can arise, for example, in battlefield surveillance, where the red points represent friendly forces, and the blue points represent enemy forces. The goal is to deploy sensors to monitor the enemy forces while avoiding the friendly forces. Similarly, the problem can arise in wildlife conservation, where the blue points represent areas of high animal activity, and the red points represent protected or private residential areas. The goal is to deploy sensors to monitor animal activity while avoiding private or protected areas. 





\begin{claim}\label{claim1}
    The algorithm of Theorem \ref{thm11} will eventually find an optimal solution (i.e., selects at most $k$ facility locations on $\ell$ to cover all the blue points) for the \textsc{AllBlue-MinRed} problem.
\end{claim}
\begin{proof}
Consider an instance of \textsc{CSofl} with $w_i=\delta$ for every $p_i\in {\cal R}$ and the weight $w_i>-|{\cal R}|\delta$ for every $p_i\in {\cal B}$, where $\delta\in \mathbb{R}$ is an arbitrary negative real value. 

\textbf{Feasibility:} The weight assignment of $w_i$ ($ > -|{\cal R}|\delta$) guarantees that all blue points will be covered. In the worst-case scenario, a single disk can cover all points, both blue and red points, whose total weight is positive due to the weight assignment. The remaining $k-1$ disks can be centered on $\ell$ to cover none of the points. This ensures that a feasible solution exists for the \textsc{AllBlue-MinRed} problem.

\textbf{Optimality:} Suppose there is a feasible solution for the \textsc{CSofl} problem that places at most $k$ facility centers on $\ell$. Let these facilities cover all points in ${\cal B}$ and some points in ${\cal R}$, with total weight equal to $\rho$. Now observe that it is impossible to improve the weight $\rho$ to $\rho'$ ($> \rho$) by relocating one of the center locations, which then uncovers $m'$ red points and one blue point (whose weight is, say, $-|{\cal R}|\delta+\epsilon$ for some $\epsilon>0$). If we do so, then the updated weight would be $\rho'=\rho-m'\delta+|{\cal R}|\delta-\epsilon$. But, $\rho'$ is no better than the earlier weight $\rho$ since $m'\leq |{\cal R}|$, $\delta<0$ and $(-m'\delta+|{\cal R}|\delta-\epsilon)<0$. Further, the optimal solution with total weight $\rho$ for the \textsc{CSofl} (computed by using the algorithm of Theorem \ref{thm11}) is also optimal for this particular variant since $\rho$ can not be improved by uncovering only red points. Hence, the above proposed algorithm for the \textsc{CSofl} problem will also correctly solve the \textsc{AllBlue-MinRed} problem.
\end{proof}


\begin{claim}\label{claim2}
The algorithm of Theorem \ref{thm11} solves the \textsc{MaxBlue-NoRed} problem optimally.
\end{claim}
\begin{proof}
Consider an instance of the \textsc{CSofl} problem, in which every $p_i\in {\cal B}$ is associated with the weight $w_i=\delta$, and every $p_i\in {\cal R}$ is associated with the weight $w_i<-|{\cal B}|\delta$, where $\delta\in \mathbb{R}$ and $\delta>0$.

\textbf{Feasibility:} Consider the following trivial feasible solution for the \textsc{CSofl} problem. Let us place $k$ facility center locations on $\ell$ so that they don't cover any blue points. Further, we reduce their radius so that no red points lie in the interior. Note that the total weight of the demand points covered by these facilities is zero. Hence, these $k$ center locations form a feasible solution for the \textsc{MaxBlue-NoRed} problem since none of the red points are covered, and the total weight is zero.

\textbf{Optimality:} Suppose we have a feasible solution with a total weight $\rho$ for the \textsc{CSofl} problem. We will try to increase this weight by relocating one of the centers covering additional $n'$ blue points and (at least) one red point. The update weight would be $\rho'=\rho+n'\delta-|{\cal B}|\delta$ which is smaller than $\rho$ since $n'\leq |{\cal B}|$ and $\delta>0$. Hence, we cannot improve the total weight by perturbing some centers to cover one more red point with the hope that it may allow us to cover some more (or even all) blue points. When we have an optimal solution with the total weight $\rho$ for an instance of the \textsc{CSofl} problem, the weight the covered blue points contribute can not be increased due to its optimality. On the other hand, this solution also can not cover any red point
because we can get a better weight $\rho'=\rho-n'\delta+|{\cal B}|\delta$ (by reducing the radius to uncover these red points. While doing so, we possibly uncover some blue points, say $n'$.) This would contradict that $\rho$ is the optimum. Hence, the above proposed algorithm (of Theorem \ref{thm11}) for the \textsc{CSofl} problem will also correctly solve the \textsc{MaxBlue-NoRed} problem.
\end{proof}

    

\begin{corollary}
The \textsc{AllBlue-MinRed} and \textsc{MaxBlue-NoRed} problems  can be solved in $O(n^3k\cdot\max{(\log{n},k)})$ time.
\end{corollary}
\begin{proof}
    Since there are only two types of weights (namely, $\delta$ and $|{\cal B}|\delta$ or $-|{\cal R}|\delta$), instead of answering circular range reporting queries, we answer range counting queries \cite{chan2016}, viz. blue count and red count for each of $O(nk)$ query circles of every candidate radius. Hence the total running time is $n^2\cdot(O(nk^2)+O(nk\log{n}))=O(n^3k\cdot \max{(\log{n},k)})$. Hence, the theorem follows from Theorem \ref{thm11} and Claims \ref{claim1} and \ref{claim2}.
\end{proof}

\subsection{The \textsc{MaxBlue-NoRed} problem for $k=1$}
In this section, we address the problem of determining the minimum enclosing disk with center on $\ell$, which encloses the maximum number of blue points without enclosing any red points. 
Recall that in the \textsc{MaxBlue-NoRed} problem, we are given two sets of points, blue points ${\cal B}$ and red points ${\cal R}$, lying above a horizontal line $\ell$, where $|{\cal B}|+|{\cal R}|=n$, the goal is to compute a minimum enclosing disk that maximizes the count of blue points (being enclosed in that disk) while ensuring that no red point is enclosed.
\begin{observation}
   If the perpendicular bisector of any two points $p_i$ and $p_j$ intersects $\ell$ at $c_i$, then there exists a disk centered at $c_i$ which has $p_i$ and $p_j$ on its boundary.
\end{observation}

\noindent The method for solving the problem is as follows:
\begin{itemize}
\item For each pair of points in the set ${\cal B}\cup {\cal R}$, compute the perpendicular bisector of the line segment connecting them. Store the intersection points of these perpendicular bisectors with the line $\ell$ in a set $I$. Also, add to the set $I$ all intersection points of $\ell$ with a vertical line through each blue point since only one blue point may also lie on the boundary of a disk in the optimal solution. 
\item For each $p_i\in I$, construct a disk centered at $p_i$ that passes through the pair of points for which the perpendicular bisector was computed in the previous step.
\item For each disk centered at $p_i\in I$, determine whether it contains any point from the set ${\cal R}$. If so, remove $p_i$ from the set $I$. Otherwise, compute the number of blue points contained in the disk.
\item If $|I|=0$, then there exists no feasible solution. Otherwise, among the disks centered at points in $I$, select the one that contains the maximum number of blue points.
\end{itemize}
\begin{figure}[!htb]
    \centering
    \includegraphics[width=0.75\textwidth]{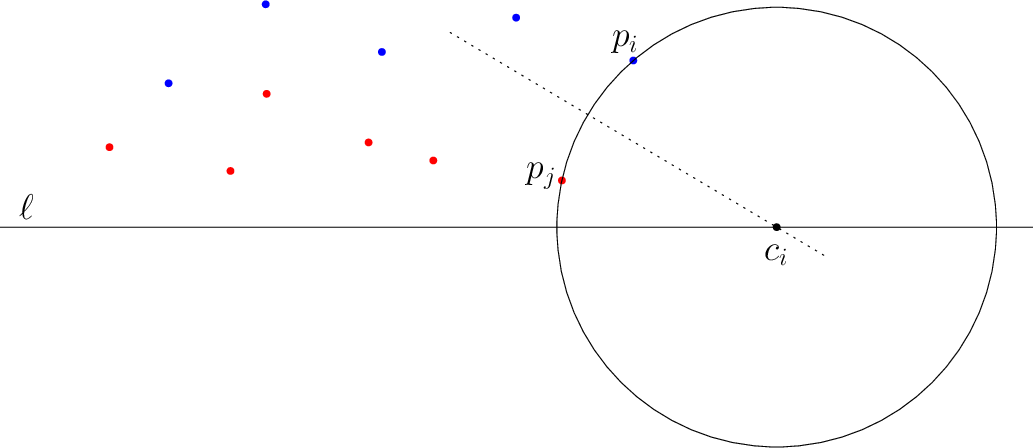}
    \caption{The optimal solution with only one blue point for $k=1$.}
    \label{fig:k1}
\end{figure}
\begin{theorem}\label{mb}
    The \textsc{MaxBlue-NoRed} problem for $k=1$ can be solved in $O(n^3)$ time.
\end{theorem}
\begin{proof}
In order to compute $I$, it requires $O(n^2)$ time. If all points in the set ${\cal B}\cup {\cal R}$ have distinct $x$-coordinates, then $|I|=O(n^2)$. Next, for each point $p_i\in I$, the time required to check the interiority of points is $O(n)$. Therefore, the total time complexity of the algorithm is $O(n^2)+O(n^3)=O(n^3)$.
\end{proof}

\noindent \textbf{Improved algorithm:} Here, we improve the running time of the algorithm of Theorem \ref{mb} by almost a linear factor. Let us recall the notations, for a point $p$, we denote its coordinates by $(x_p,y_p)$, and for a pair of points $p,q$, we let $C_{p,q}$ denote the circle whose center lies on the line $y=0$ and whose boundary passes through $p,q$, and we let $(x_{p,q},0)$ denote the center of $C_{p,q}$. Let $C_p$ be a circle with its boundary passing through a point $p$, and its center lying on $\ell$ such that its radius equals $y_p$ and its center at the coordinates $(x_p,0)$.

\begin{claim}\label{clm:pqr}
Given three points $p,q,r$, the point $r$ lies on or inside $C_{p,q}$ if and only if one of the following is true:

(i) $x_p<x_r$ and $x_{p,q} \geq x_{p,r}$;

(OR)

(ii) $x_p>x_r$ and $x_{p,q} \leq x_{p,r}$.
\end{claim}

\begin{proposition}
There is an algorithm that accepts two sets ${\cal B},{\cal R}$ of $n$ ($|{\cal B}|+|{\cal R}|$) points on the plane and finds for every pair $p,q$ of points in ${\cal B}$, the number of points of ${\cal R}$ that lie on or inside the circle $C_{p,q}$.
Further, this algorithm runs in time $O(n^2\log n)$.
\end{proposition}

\begin{proof}
We first describe the algorithm.
\begin{enumerate}
	\item Sort the point sets ${\cal B}\cup{\cal R}$ based on their $x$-coordinates from left to right.
       \item For each $p \in {\cal B}$, compute three lists: $C_{p,{\cal B}}=\{x_{p,q}|q \in {\cal B} \setminus \{p\}\}$, $C_{p,{\cal R},1}=\{x_{p,r}|r \in {\cal R} \text{ and } x_p<x_r\}$,
	$C_{p,{\cal R},2}=\{x_{p,r}|r \in {\cal R} \text{ and } x_p>x_r\}$.
	\item For each $p$, sort the lists $L_{p,1}=C_{p,{\cal B}} \cup C_{p,{\cal R},1}$
	and $L_{p,2}=C_{p,{\cal B}} \cup C_{p,{\cal R},2}$.
	\item For each $p$, do the following:
	by making a single pass over $L_{p,1}$, compute for every $q \in {\cal B} \setminus \{p\}$, the value $N_{p,q,1}$, which is defined to be the number of elements of $C_{p,{\cal R},1}$ that appear before $x_{p,q}$ in $L_{p,1}$.
	\item For each $p$, compute for every $q \in {\cal B} \setminus \{p\}$, the value $N_{p,q,2}$, which is defined to be the number of elements of $C_{p,{\cal R},2}$ that appear after $x_{p,q}$ in $L_{p,2}$.
	\item For each $p,q$, the desired count (i.e., the number of red points covered by the disk $C_{p,q}$) is $N_{p,q,1}+N_{p,q,2}$. 
  \item To examine the scenario where only a single blue point resides on the circle, we construct a list in the following manner:
        \begin{itemize}
            \item For each $p\in {\cal B}$, we select an arbitrary point $p_{temp}$ that lies on the circle $C_p$.
            \item Let $C_{\cal B}=\{(p,p_{temp})|\text{ $p\in {\cal B}$ and $p_{temp}$ is an arbitrary point lying on $C_p$}\}$ be a list.
            \item Now, assign $C_{p,{\cal B}}=\{x_{p,q}|(p,q)\in C_{\cal B}\}$ in step 2 and compute the lists $C_{p,{\cal R},1}$ and $C_{p,{\cal R},2}$, then repeat the remaining steps till step 6.
        \end{itemize}
    \item Lastly, we determine the circle that encloses the maximum number of blue points and none of the red points by answering circular range counting queries for every circle $C_{p,q}$ which has the red count $N_{p,q,1}+N_{p,q,2}=0$
\end{enumerate}

\noindent{\bf Analysis:}
The correctness follows from Claim \ref{clm:pqr}.
The running time is dominated by steps 3, 4, 5, and 8.
Step 3 takes time $O(n \log n )$ for a single point $p$ and hence total time $O(n^2 \log n)$; steps 4 and 5 take time $O(n^2)$ each. The time complexity of Step 8 is $O(n^2\log{n})$ due to the repetition of the algorithm to determine the maximum number of blue points (i.e., the value $N_{p,q,1}+N_{p,q,2}$ is maximum for the blue points) enclosed by each circle $C_{p,q}$ satisfying $N_{p,q,1}+N_{p,q,2}=0$ for the points in ${\cal R}$.

\end{proof}

\begin{theorem}
    The \textsc{MaxBlue-NoRed} problem for $k=1$ can be solved in $O(n^2\log{n})$ time.
\end{theorem}
\subsection{The \textsc{AllBlue-MinRed} problem for $k=1$}
In \cite{bitner2010}, the problem of finding the largest and smallest disk that covers all blue points and as few red points as possible is studied in its unrestricted version, and in \cite{cheung2010}, they gave linear time (expected) algorithm based on linear programming. On the other hand, in \cite{bereg2015}, the problem is studied when the center of the disk is restricted to a line segment, which has the same time complexity as our farthest-point Voronoi diagram based algorithm. This problem has many bichromatic variants studied in the Ph.D. thesis \cite{armaselu2017t}. Next we describe our algorithm which is based on the farthest-point Voronoi diagram.


We first construct a farthest point Voronoi diagram ${\cal FVD}$ for the set of blue points. Then, we find all the intersection points of Voronoi edges with $\ell$. Since every Voronoi edge corresponds to the farthest pair of two points, we can determine a disk centered at an intersection point of the Voronoi edge and $\ell$, such that the disk's boundary passes through the respective two points. We find all such disks for each intersection point, which are candidate locations for a facility in any feasible solution to the \textsc{AllBlue-MinRed} problem for $k=1$. We pick a disk that covers the minimum number of red points among these disks. To this end, we again employ range searching algorithms of \cite{chan2016}.

\begin{figure}[!ht]
	\begin{centering}
		\includegraphics[width=\textwidth]{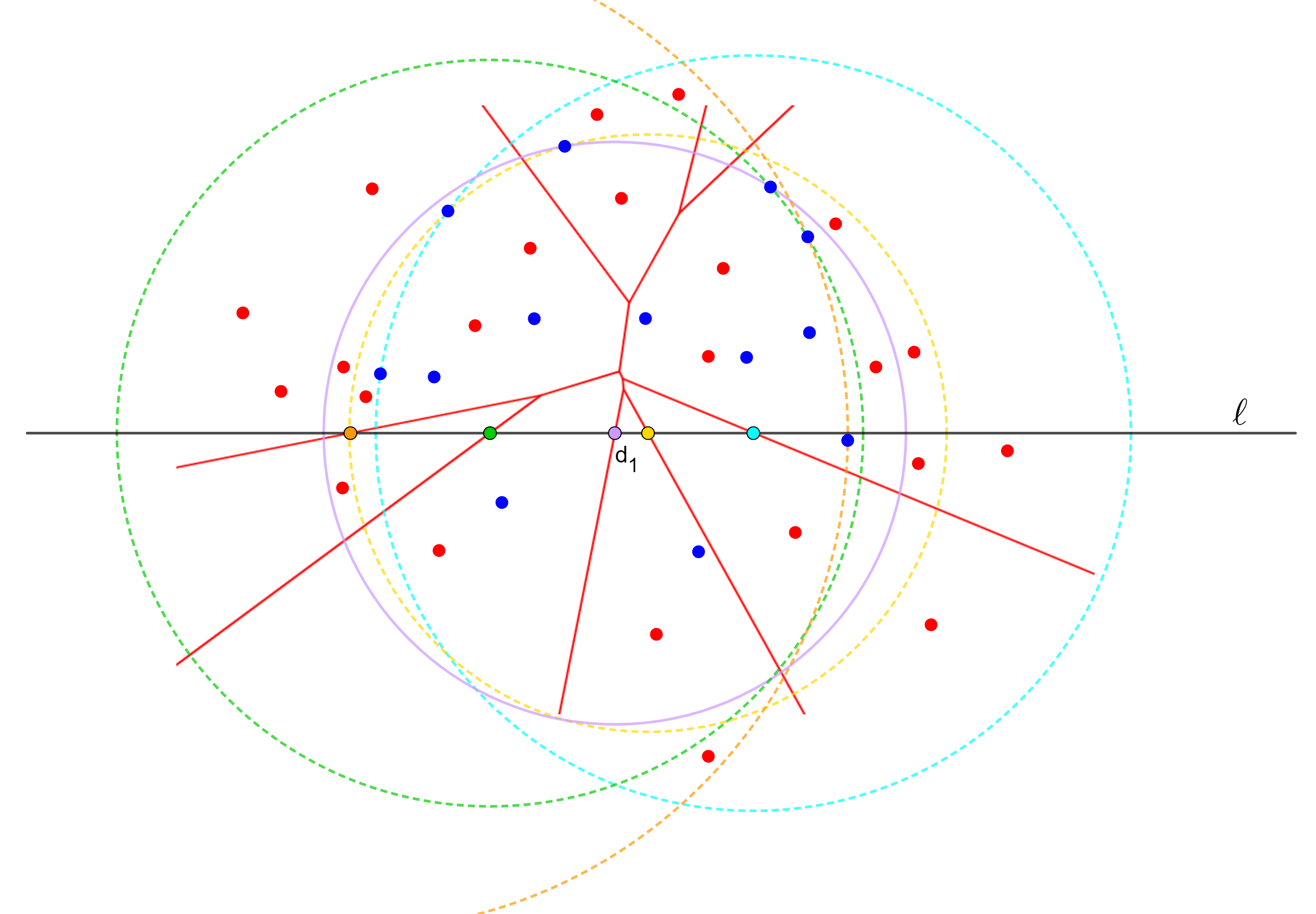}
		\caption{The pink disk $d_1$ covers all blue points with a fewer number of red points.}
		\label{figb1}
	\end{centering}
\end{figure}

\begin{theorem}
	We can solve the \textsc{AllBlue-MinRed} problem for $k=1$ in $O(n\log{n})$ time.
\end{theorem}
\begin{proof}
The construction of the ${\cal FVD}$ concerning the blue points will take $O(n\log{n})$ time. There will be, at most, $n$ Voronoi edges, which will intersect with the line $\ell$. We need to find a disk among the disks centered at all candidate locations, as identified above. This disk should cover the minimum number of red points. To do this, we first pre-process the red points in $O(n\log{n})$ time \cite{chan2016}. Each query corresponding to a disk centered at a candidate location will take $O(\log{n})$ time \cite{chan2016}. Hence, the total running time is $O(n\log{n})$.
The algorithm's correctness follows because the Voronoi edge on which we center the disk will correspond to the farthest pair of blue points, and the disk whose boundary passes through this pair will cover all the blue points.
\end{proof}

\section{\textsc{Sofl} on $t$-lines}
Let us consider a given set ${\cal B}$ of blue points and a set ${\cal R}$ of red points, which are positioned around $t$ parallel lines denoted as ${\ell_1,\ell_2,\dots,\ell_t}$ in the plane. These lines may have arbitrary vertical displacements. Each point $p_i\in {\cal B}\cup{\cal R}$ is assigned a weight denoted as $w_i$, where $w_i>0$ if $p_i\in {\cal B}$ and $w_i<0$ if $p_i\in {\cal R}$. The cardinality of the set ${\cal B}\cup{\cal R}$ is denoted as $n$, and the interior of any geometric object $d$ is represented as $d^0$ (excluding its boundary $\partial d$).

The objective is to pack $k$ non-overlapping congruent disks, denoted as $d_1$, $d_2$, $\ldots$, $d_k$, with the smallest possible radius. These disks must be centered on the parallel lines closest to the points covered by each disk. The goal is to maximize the sum of the weights of the points covered by the interior of the disks. This sum is represented as $\sum\limits_{j=1}^k\sum\limits_{{i:\exists p_i\in{\cal R}, p_i\in d_j^0}}w_i+\sum\limits_{j=1}^k\sum\limits_{{i:\exists p_i\in{\cal B}, p_i\in d_j}}w_i$.

We may consider this as a generalization of the \textsc{Sofl} problem, where $t$ horizontal lines are present, and facilities can be centered on any of these lines. Following a similar approach as in Section \ref{sec3}, we obtain all the candidate radii ${\cal L}_{\textsc{can}}$ independently for each of the $t$ lines and let us denote it as ${\cal L}_{\textsc{tcan}}$. Note that the cardinality of ${\cal L}_{\textsc{tcan}}$ is $O(tn^2)$. Hence we have the following lemma.
\begin{lemma}\label{tslemma1}
    $|{\cal L}_{\textsc{tcan}}|=O(tn^2)$
\end{lemma}

Next, we fix a radius $r_{\textsc{can}}\in {\cal L}_{\textsc{tcan}}$. We can transform the problem into finding the minimum weight $k$-link path in a directed acyclic graph (DAG) $G(V', E')$, as discussed in Section \ref{sec4}. However, the cardinality of the set $V'$ is $O(nkt^2)$, since each point $p_i\in {\cal B}\cup {\cal R}$ can create an influence interval on each of the $t$ lines, resulting in $O(nt)$ endpoints of the influence intervals and adding $O(kt)$ additional points (see Figure \ref{tlines}). Figure \ref{tlines} depicts the candidate locations on $\ell_{i+1}$ and $\ell_{i-1}$, located at a distance of $2\lambda$ to the right of $p_{\ell_i}$. Similarly, the mirror case can be considered for the point situated at a distance of $2\lambda$ to the left of $p_{\ell_i}$ on $\ell_{i+1}$ and $\ell_{i-1}$.
\begin{figure}[!htb]
	\begin{centering}
		\includegraphics[scale=0.8]{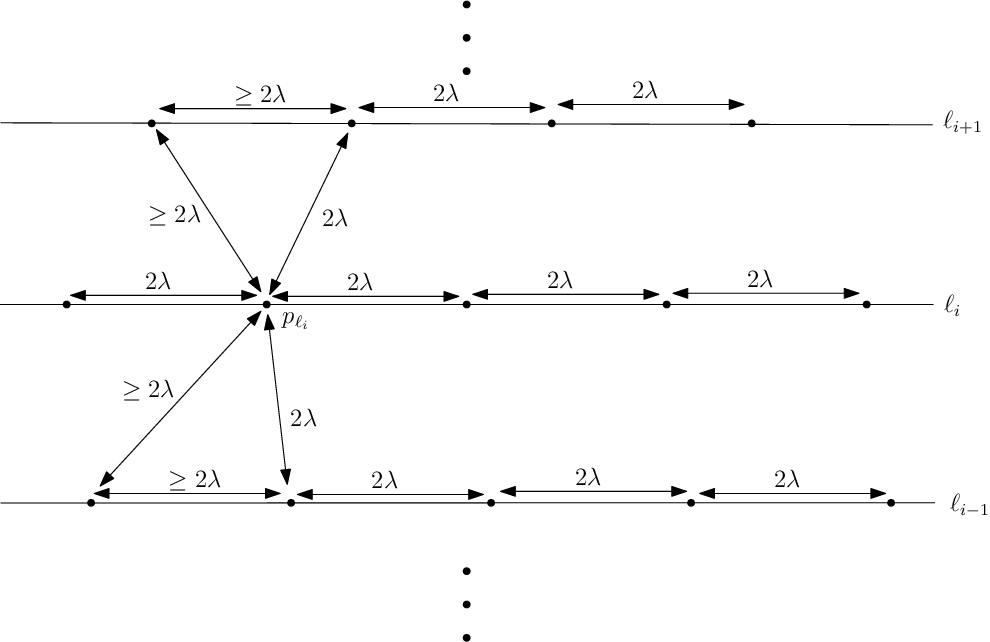}
		\caption{Candidate locations corresponding to the endpoint of the infeasible region $p_{\ell_i}$ and candidate radius $\lambda$.}
		\label{tlines}
	\end{centering}
\end{figure}

Without loss of generality, we assume that all the points in $V'$ have distinct $x$-coordinates. We can construct $G'$ in $O(n^2k^2t^4)$ time by employing the sweeping technique, specifically sweeping from left to right. Therefore, the following lemma holds.
\begin{lemma}\label{tslemma2}
      The DAG $G'$ on $t$-lines can be constructed in $O(n^2k^2t^4)$ time.  
\end{lemma}
\begin{proof}
Follows from the Lemma \ref{lemma8} as the cardinalities $|V'|=O(nkt^2)$ and $|E'|=O(n^2k^2t^4)$.
\end{proof}

\begin{theorem}
    The \textsc{Sofl} problem of $t$-lines can be solved exactly in $O(n^4k^2t^5)$ time.
\end{theorem}
\begin{proof}
    Follows from the Lemma \ref{tslemma1} and Lemma \ref{tslemma2} since there are $O(n^2t)$ candidate radii and the total time is $O(n^2t)\times O(n^2k^2t^4)=O(n^4k^2t^5)$.
\end{proof}

\section{Discrete \textsc{Sofl} with all facility sites in convex position}
Suppose we are given a set $\cal B$ of blue points, a set $\cal R$ of red points, and a set ${\cal F}$ of $s$ candidate locations in convex position; all these three sets are in the plane. Let the weight of a given point $p_i\in {\cal B}\cup{\cal R}$ be $w_i>0$ if $p_i\in {\cal B}$ and $w_i<0$ if $p_i\in {\cal R}$, $|{\cal B}\cup{\cal R}|=n$, and $d^0$($=d\setminus\partial d$) be the interior of any geometric object $d$. We wish to pack $k$ non-overlapping congruent disks $d_1$, $d_2$, \ldots, $d_k$ of minimum radius, centered at points in ${\cal F}$ such that $\sum\limits_{j=1}^k\sum\limits_{\{i:\exists p_i\in{\cal R}, p_i\in d_j^0\}}w_i+\sum\limits_{j=1}^k\sum\limits_{\{i:\exists p_i\in{\cal B}, p_i\in d_j\}}w_i$ is maximized, i.e., the sum of the weights of the points covered by $\bigcup\limits_{j=1}^kd_j$ is maximized.

The above problem is a discrete variation of the \textsc{Sofl} problem (\textsc{Dsofl}) because a finite number of candidate facility sites (in convex position) are pre-given. Even though it is the discrete version of the \textsc{Sofl} problem, similar to the continuous line case, we know that there exists only a constant number of critical configuration types for the points in $\mathcal{R} \cup \mathcal{B}$ and candidate facilities in $\mathcal{F}$. It follows from the latter that we also have a finite number of candidate radii here. Let $\mathcal{L}_{\textsc{Dcan}}$ denote the set of all candidate radii.

\begin{lemma}\label{dslemma1}
$|\mathcal{L}_{\textsc{Dcan}}|=O(ns)$.
\end{lemma}

\begin{proof}
Since there will be only a constant number of critical configuration types concerning points $\mathcal{B} \cup \mathcal{R}$ and candidate facilities $\mathcal{F}$, we can consider the following situation where the candidate radius is determined based on a point in $\mathcal{B} \cup \mathcal{R}$ and a candidate facility location (on whose boundary that point lies) in $\mathcal{F}$. The cardinality of the set of radii from this situation is $O(ns)$ (see Figure \ref{DSk4}). 

\begin{figure}[!htb]
	\begin{centering}
		\includegraphics[scale=0.6]{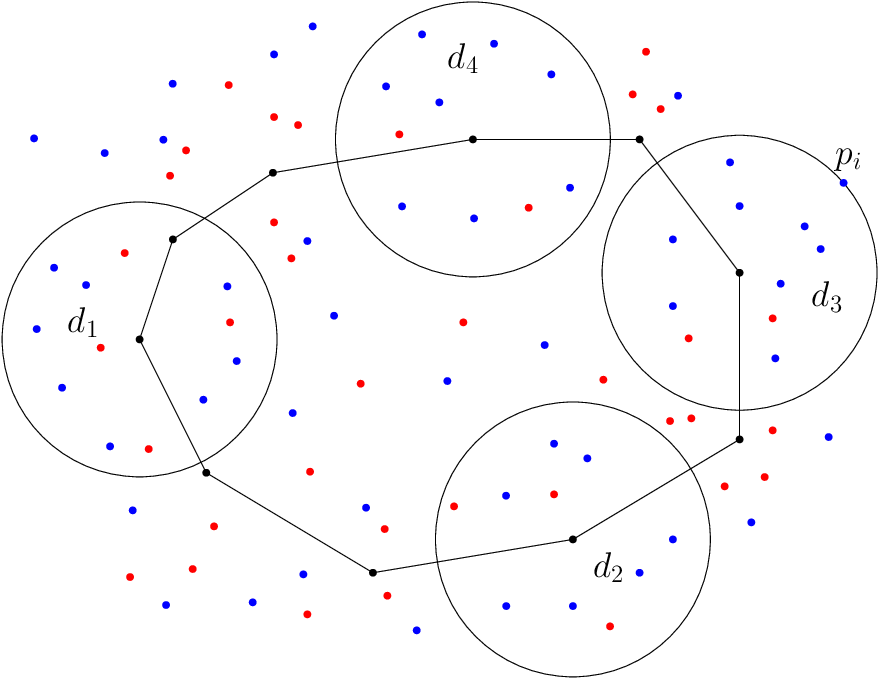}
		\caption{The blue point $p_i$ lying on the boundary of $d_3$ will determine the radius.}
		\label{DSk4}
	\end{centering}
\end{figure}

\begin{figure}[!htb]
	\begin{centering}
		\includegraphics[scale=0.6]{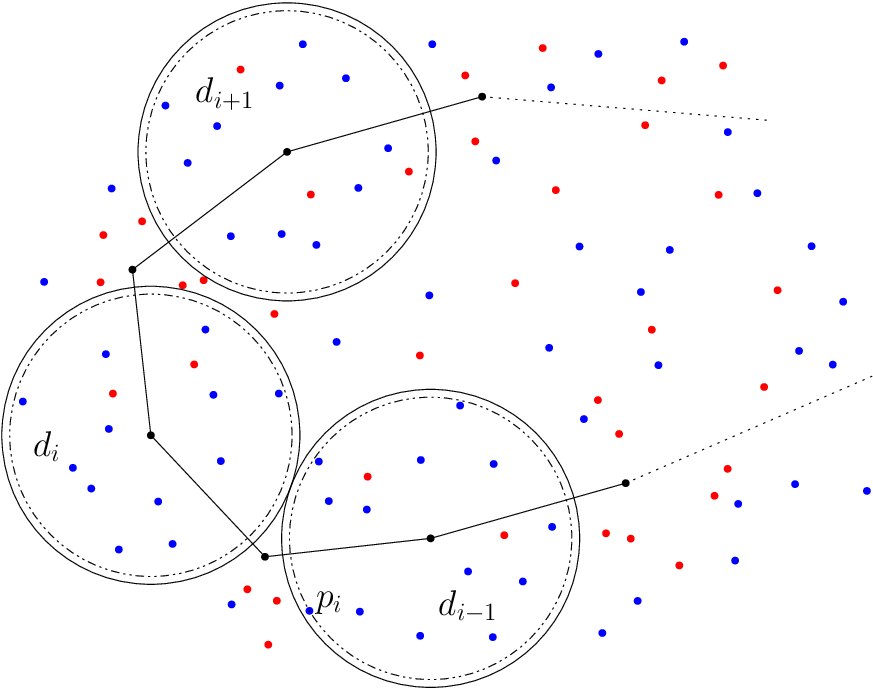}
		\caption{Illustration of the candidate facilities will not determine the radius of disks in the optimal packing.}
		\label{DSs2}
	\end{centering}
\end{figure}

The radius of the disks in the optimal packing cannot be determined solely by the distance between the candidate sites (see Fig \ref{DSs2}). In Figure \ref{DSs2}, we can observe that the closest pair of disks $d_{i-1}$ and $d_{i}$ will never touch in any optimal packing (i.e., the distance between them will not determine the radii of the disks in the optimal packing). Suppose they touch in any optimal packing, then we can reduce the radii of the disks until one of the blue points lie on the boundary of any of the disk (see Figure \ref{DSs2}, $p_i$ lying on the boundary of the disk $d_{i-1}$).
\end{proof}

\subsection{Dynamic programming algorithm}
In this section, first, we show a relationship between the Voronoi diagram of points in an optimal solution and the cost of an optimal solution to the \textsc{Dsofl} problem. We then present a dynamic programming-based solution for the discrete \textsc{Sofl} problem utilizing this property of the Voronoi diagram ($\mathcal{VD}$) of the $k$ sites in an optimal solution. This process is repeated for each $\lambda \in \mathcal{L}_{\textsc{Dcan}}$. The Voronoi diagram of points in convex position forms a tree-like structure except for its infinite edges. If we overlay the diagram with a sufficiently big bounding rectangle, we have the following observation.

\begin{observation}\label{obd}
The Voronoi diagram of points in convex position is a tree.
\end{observation}

Since the points in $\mathcal{F}$ are in convex position, Observation \ref{obd} implies that the $\mathcal{VD}$ of any subset of points in $\mathcal{F}$ is also a tree. Hence, $\mathcal{VD}$ of the optimal $k$ facility sites is a tree. This $\mathcal{VD}$ tree structure allows us to employ dynamic programming. To find this tree or a subtree of it from its rightmost node, we define a subproblem that explores all possible edges from the rightmost node and then further this exploration recursively. This leads to recursively constructing an optimal solution once we guess the rightmost node of the tree. Furthermore, we show no circular dependencies between subproblems.

Without loss of generality, let us consider that the points in ${\cal F}=\{p_1,p_2,\dots,p_s\}$ are ordered clockwise. It is known that the Delaunay triangulation is the dual of the Voronoi diagram. Denote ${\cal DT}$ as the Delaunay triangulation formed by the points corresponding to the Voronoi centers in the Voronoi diagram, ${\cal VD}$. Observe that the smallest edge length of ${\cal DT}$ of points in an optimal solution to \textsc{Dsofl} is at least twice the radius of disks in the optimal solution. 

For a given $\lambda \in \mathcal{L}_{\textsc{Dcan}}$, we precalculate the weight of points covered by a disk with radius $\lambda$ centered at a facility site $f_i\in {\cal F}$ and denote this weight as $w(f_i)$. Then, our dynamic program-based algorithm is as follows. First, we guess the Delaunay triangle corresponding to the rightmost Voronoi node, the three rightmost facility centers, say, $p_i, p_\ell, p_j,$ (where $p_l$ is the rightmost and $p_i$ is below $p_j$) in the optimal solution. We make all possible $s\choose 3$ guesses to find these three optimal centers.
Then, define a subproblem $\Gamma(p_i,p_\ell, p_j, {\cal F'}; {\cal K})$, which corresponds to the maximum (optimal) weight of the points covered by ${\cal K}$ facilities located at some points in $\mathcal{F}$ with a radius of $\lambda$, and the points $p_i$, $p_\ell$ and $p_j$ are the rightmost ordered points in the optimal solution. Initially, we set ${\cal K}=k-3$ and ${\cal F'}={\cal F}\setminus \{p_i,\dots,p_\ell, \dots,p_j\}$, where the indices
$i, j, \ell, \ell'$ are to be read modulo $s$.
Now, consider reconstructing ${\cal VD}$ with three $p_i$, $p_j, p_\ell$ fixed on the right. We do this by determining the corresponding ${\cal DT}$ triangle with its corner points $p_i$, $p_j$ and $p_{\ell'}$. We extend the ${\cal VD}$ by choosing the next point $p_{\ell'}$ that lies left of $ \overrightarrow{p_jp_i}$ such that it is at least $2\lambda$ from $p_i$, $p_j$ and $p_{\ell}$. Observe that $p_{\ell'}$ is outside the circumcircle of $p_i$, $p_j$ and $p_{\ell}$. Then we have the following recurrence,
\vspace{1em}

    \noindent $\Gamma(p_i,p_\ell, p_j, {\cal F'};{\cal K}) = $
\begin{equation*}
\begin{split}
  \max_{\substack{{\cal K}'\leq {\cal K}-1, \\  p_{\ell'}\in {\cal F'} and\\ p_{\ell'}: \zeta(p_{\ell'},\{p_i,p_j,p_\ell\})\geq 2\lambda}}\begin{cases}
     w(p_{\ell'})+\Gamma(p_{\ell'},p_i,p_j, {\cal F''}; {\cal K}')+\Gamma(p_{\ell'},p_j,p_i, {\cal F'''}; {\cal K}-1-{\cal K}') & \text{ if } |{\cal F'}|\geq 1 \\
     0 & \text{ otherwise}
     \end{cases}
\end{split}
\end{equation*}
 \hspace{1em}where $\zeta(p_{\ell'},\{p_i,p_j,p_\ell\})=\min\{dist(p_{\ell'},p_i),dist(p_{\ell'},p_j),dist(p_{\ell'},p_{\ell})\}$, $w(p_{\ell'})$ denotes the total weight of the points that are covered by a disk of radius $\lambda$ centered at $p_{\ell'}$, ${\cal F''}={\cal F'}\setminus \{p_j,\dots,p_{\ell'}\}$ and ${\cal F'''}={\cal F'}\setminus \{p_{\ell'},\dots,p_i\}$.
Base cases are $\Gamma(p_i, p_\ell, p_j, {\cal F'}; 0) = w(p_i)+w(p_{\ell})+w(p_j)$, $\Gamma(p_i, p_\ell, p_j, \emptyset; {\cal K}) = 0$.

\subsubsection*{Proof of Correctness:}
The correctness of the dynamic programming algorithm can be established based on the following observations:
\begin{itemize}
    \item The minimum edge length of ${\cal DT}$ formed by the $k$ sites in the optimal solution is at least $2\lambda$. This ensures that the disks in the optimal solution do not overlap, as the distance between any two points in the solution is greater than or equal to $2\lambda$.
    \item There always exists a solution for a given set of points ${\cal B}\cup {\cal R}$ and ${\cal F}$. If it is impossible to place $k$ disks with a radius of $\lambda$, the algorithm returns a zero weight, indicating that a solution does not exist for given $\lambda$.
    \item By assuming that $p_i$, $p_\ell$, and $p_j$ are the rightmost points in the optimal solution, we have $O(s^3)$ choices for these points. This assumption ensures that a ${\cal DT}$ with $k$ vertices corresponding to the optimal solution always exists if a solution exists for a given $\lambda$.
\end{itemize}
Based on these observations, we can conclude that the dynamic programming algorithm is correct in determining the optimal solution for the given set of points ${\cal B}\cup {\cal R}$ and ${\cal F}$, considering the assumptions made and the properties of the ${\cal DT}$ formed by the candidate sites.

 \begin{theorem}
     Discrete \textsc{Sofl} with candidate facility sites in convex position can be solved in polynomial time. 
 \end{theorem}
 \begin{proof} The running time of the algorithm is calculated as follows:
     \begin{itemize}
         \item From Lemma \ref{dslemma1} we have $|\mathcal{L}_{\textsc{Dcan}}|=O(ns)$.
         \item For each $\lambda\in {\cal L}_{\textsc{tcan}}$ we call the dynamic programming algorithm.
         \item Dynamic programming algorithm for a given $\lambda$:
         \begin{itemize}
             \item For each $f_i\in {\cal F}$, calculating weight of points covered by a disk of radius $\lambda$ centered at $f_i$ will take $(ns)$ time.
             \item There are $O(s^3k)$ subproblems and each subproblem will take $O(sk)$ time.
         \end{itemize}
         \item The total time complexity of the algorithm is $O(n^2s^2+ns^5k^2)$. Additionally, we designate the vertices of ${\cal DT}$ as the optimal solution that yields the maximum weight out of all the invocations of the dynamic programming algorithm with three rightmost points $p_i$, $p_j$, $p_{\ell}$.
     \end{itemize}
 \end{proof}

\section{Conclusion}




This paper studied the problem of locating $k$ semi-obnoxious facilities constrained to a line (\textsc{CSofl}) when the given demand points have positive and negative weights. Specifically, we solved the problem of locating $k$ semi-obnoxious facilities on a line to locate facilities with the maximum weight of the covered demand points in $O(n^4k^2)$ time. Subsequently, we improved the running time to $O(n^3k \cdot\max{(n, k)})$. Furthermore, we addressed two special cases of the problem where points do not have arbitrary weights. We showed that these two special cases can be solved in $O(n^3k\cdot\max{(\log{n}, k)})$ time. For the first case, when $k=1$, we also provide an algorithm that solves the problem in $O(n^3)$ time, and subsequently, we improve this result to $O(n^2\log{n})$. For the latter case, we give $O(n\log{n})$ time algorithm that uses the farthest point Voronoi diagram. We also studied the \textsc{Sofl} for $t$-lines and showed that it can be solved in polynomial time but with a high order degree in $t$. Further, we investigated the complexity of discrete semi-obnoxious facility location (\textsc{DSofl}) for the given candidate locations in convex position, and we showed that this problem can also be solved in polynomial time.  

Following are some of the open problems that are worth considering as future work:
\begin{itemize}
    \item The continuous unrestricted variant of the semi-obnoxious facility location problem: given two sets (red and blue) of demand points with positive and negative weights (respectively) in the plane and an integer $k$. The objective is to maximize the sum of the weights of the points covered by the union of $k$ congruent non-overlapping disks of minimum radius centered anywhere in the plane (i.e., disks (facilities) may be centered anywhere in the plane).

    \item The discrete unrestricted variant of the semi-obnoxious facility location problem: given two sets (red and blue) of points with positive and negative weights (respectively), a set of candidate facility locations in the plane, and an integer $k$. The objective is to maximize the sum of the weights of the points covered by the union of $k$ congruent non-overlapping disks of minimum radius centered at some of the candidate facility locations.
    \item To investigate the scenario when the disks are centered on the boundary of a convex polygon instead of a horizontal line or at vertices of convex polygon.
    \item To investigate the scenario when the disks are restricted to be centered at the grid points of a $t\times t$ grid in the plane. 
    
    \item Finding a better than $O(n^2\log n)$ time algorithm for the \textsc{MaxBlue-NoRed} problem for $k=1$.

\end{itemize}

	\bibliography{refs}

\begin{thebibliography}{22}
\providecommand{\natexlab}[1]{#1}
\providecommand{\url}[1]{\texttt{#1}}
\expandafter\ifx\csname urlstyle\endcsname\relax
  \providecommand{\doi}[1]{doi: #1}\else
  \providecommand{\doi}{doi: \begingroup \urlstyle{rm}\Url}\fi

\bibitem[Abidha and Ashok(2022)]{abidha2022}
VP~Abidha and Pradeesha Ashok.
\newblock Geometric separability using orthogonal objects.
\newblock \emph{Information Processing Letters}, 176:\penalty0 106245, 2022.

\bibitem[Aggarwal et~al.(1993)Aggarwal, Schieber, and Tokuyama]{agga1993}
Alok Aggarwal, Baruch Schieber, and Takashi Tokuyama.
\newblock Finding a minimum weight k-link path in graphs with monge property
  and applications.
\newblock In \emph{Proceedings of the ninth annual symposium on Computational
  geometry}, pages 189--197, 1993.

\bibitem[Armaselu(2017)]{armaselu2017t}
Bogdan~Andrei Armaselu.
\newblock \emph{On the geometric separability of bichromatic point sets}.
\newblock PhD thesis, 2017.

\bibitem[Bereg et~al.(2015)Bereg, Daescu, Zivanic, and Rozario]{bereg2015}
Sergey Bereg, Ovidiu Daescu, Marko Zivanic, and Timothy Rozario.
\newblock Smallest maximum-weight circle for weighted points in the plane.
\newblock In \emph{Computational Science and Its Applications--ICCSA 2015: 15th
  International Conference, Banff, AB, Canada, June 22-25, 2015, Proceedings,
  Part II}, pages 244--253. Springer, 2015.

\bibitem[Bitner et~al.(2010)Bitner, Cheung, and Daescu]{bitner2010}
Steven Bitner, Yam Cheung, and Ovidiu Daescu.
\newblock Minimum separating circle for bichromatic points in the plane.
\newblock In \emph{2010 International Symposium on Voronoi Diagrams in Science
  and Engineering}, pages 50--55. IEEE, 2010.

\bibitem[Chan and Tsakalidis(2016)]{chan2016}
Timothy~M Chan and Konstantinos Tsakalidis.
\newblock Optimal deterministic algorithms for 2-d and 3-d shallow cuttings.
\newblock \emph{Discrete \& Computational Geometry}, 56:\penalty0 866--881,
  2016.

\bibitem[Cheung et~al.(2010)Cheung, Daescu, and Richardson]{cheung2010}
Yam Cheung, Ovidiu Daescu, and TX~Richardson.
\newblock Minimum separating circle for bichromatic points by linear
  programming.
\newblock \emph{Proc. FWCG}, 2010.

\bibitem[Coutinho-Rodrigues et~al.(2012)Coutinho-Rodrigues, Tralh{\~a}o, and
  Al{\c{c}}ada-Almeida]{co2013}
Jo{\~a}o Coutinho-Rodrigues, Lino Tralh{\~a}o, and Lu{\'\i}s
  Al{\c{c}}ada-Almeida.
\newblock A bi-objective modeling approach applied to an urban semi-desirable
  facility location problem.
\newblock \emph{European journal of operational research}, 223\penalty0
  (1):\penalty0 203--213, 2012.

\bibitem[Edelsbrunner and Preparata(1988)]{E1988}
H.~Edelsbrunner and F.P. Preparata.
\newblock Minimum polygonal separation.
\newblock \emph{Information and Computation}, 77\penalty0 (3):\penalty0
  218--232, 1988.
\newblock ISSN 0890-5401.
\newblock
  \href{https://doi.org/https://doi.org/10.1016/0890-5401(88)90049-1}{\ttfamily\path{
  doi:https://doi.org/10.1016/0890-5401(88)90049-1}}.

\bibitem[Fekete(1992)]{fek1992}
Sandor Fekete.
\newblock On the complexity of min-link red-blue separation.
\newblock \emph{Manuscript, department of applied mathematics, SUNY Stony
  Brook, NY}, 1992.

\bibitem[Gholami and Fathali(2021)]{gh2021}
Mehraneh Gholami and Jafar Fathali.
\newblock The semi-obnoxious minisum circle location problem with euclidean
  norm.
\newblock \emph{International Journal of Nonlinear Analysis and Applications},
  12\penalty0 (1):\penalty0 669--678, 2021.

\bibitem[Golpayegani et~al.(2014)Golpayegani, Fathali, and Khosravian]{go2014}
Mehdi Golpayegani, Jafar Fathali, and Eiman Khosravian.
\newblock Median line location problem with positive and negative weights and
  euclidean norm.
\newblock \emph{Neural Computing and Applications}, 24:\penalty0 613--619,
  2014.

\bibitem[Golpayegani et~al.(2017)Golpayegani, Fathali, and Moradi]{go2017}
Mehdi Golpayegani, Jafar Fathali, and Haleh Moradi.
\newblock A particle swarm optimization method for semi-obnoxious line location
  problem with rectilinear norm.
\newblock \emph{Computers \& Industrial Engineering}, 109:\penalty0 71--78,
  2017.

\bibitem[Maranas and Floudas(1994)]{ma1994}
Costas~D Maranas and Christodoulos~A Floudas.
\newblock A global optimization method for weber’s problem with attraction
  and repulsion.
\newblock \emph{Large scale optimization: State of the art}, pages 259--285,
  1994.

\bibitem[Melachrinoudis and Xanthopulos(2003)]{me2003}
Emanuel Melachrinoudis and Zaharias Xanthopulos.
\newblock Semi-obnoxious single facility location in euclidean space.
\newblock \emph{Computers \& Operations Research}, 30\penalty0 (14):\penalty0
  2191--2209, 2003.

\bibitem[Mitchell(1993)]{mit1993}
Joseph~SB Mitchell.
\newblock Approximation algorithms for geometric separation problems.
\newblock Technical report, State University of New York at Stony Brook, 1993.
\newblock URL \url{http://www.ams.sunysb.edu/~jsbm/papers/sep-2-10-94.pdf}.

\bibitem[O'rourke et~al.(1986)O'rourke, Rao~Kosaraju, and Megiddo]{o1986}
Joseph O'rourke, S~Rao~Kosaraju, and Nimrod Megiddo.
\newblock Computing circular separability.
\newblock \emph{Discrete \& Computational Geometry}, 1:\penalty0 105--113,
  1986.

\bibitem[Singireddy and Basappa(2021)]{vi2021}
Vishwanath~R Singireddy and Manjanna Basappa.
\newblock Constrained obnoxious facility location on a line segment.
\newblock In \emph{33rd Canadian Conference on Computational Geometry}, pages
  362--367, 2021.

\bibitem[Singireddy and Basappa(2022)]{sing2022}
Vishwanath~R Singireddy and Manjanna Basappa.
\newblock Dispersing facilities on planar segment and circle amidst repulsion.
\newblock In \emph{Algorithmics of Wireless Networks: 18th International
  Symposium on Algorithmics of Wireless Networks, ALGOSENSORS 2022, Potsdam,
  Germany, September 8--9, 2022, Proceedings}, pages 138--151. Springer, 2022.

\bibitem[Teran-Somohano and Smith(2019)]{te2019}
Alejandro Teran-Somohano and Alice~E Smith.
\newblock Locating multiple capacitated semi-obnoxious facilities using
  evolutionary strategies.
\newblock \emph{Computers \& Industrial Engineering}, 133:\penalty0 303--316,
  2019.

\bibitem[Wagner(2015)]{wa2015}
Andrea Wagner.
\newblock A new duality based approach for the problem of locating a
  semi-obnoxious facility.
\newblock 2015.

\bibitem[Zhang(2022)]{zhang2022}
Bowei Zhang.
\newblock Efficient algorithms for obnoxious facility location on a line
  segment or circle.
\newblock \emph{arXiv preprint arXiv:2210.07146}, 2022.

\end{thebibliography}
	
	\appendix
	
	
	
\end{document}